\documentclass[11pt,a4paper]{article}
\usepackage[T1]{fontenc}
\usepackage[utf8]{inputenc}
\usepackage{lmodern}
\usepackage[margin=28mm]{geometry}
\usepackage{microtype}
\usepackage{amsmath,amssymb,amsthm,mathtools}
\usepackage[authoryear,round]{natbib}
\usepackage{xurl}
\usepackage[unicode,bookmarksnumbered,hidelinks]{hyperref}
\hypersetup{pdftitle={Strategy-proof choice on strictly convex frontiers},
 pdfauthor={Siwei Chen and Pengbo Wang},
 pdfsubject={Individual incentives and collective choice on a constrained boundary},
 pdfkeywords={strategy-proofness, unanimity, public-good location, strictly convex frontiers, attainable menus}}
\newcommand{\R}{\mathbb R}

\newcommand{\Circ}{\mathbb S^1}
\newtheorem{theorem}{Theorem}
\newtheorem{proposition}{Proposition}
\newtheorem{corollary}{Corollary}
\newtheorem{lemma}{Lemma}
\theoremstyle{definition}
\newtheorem{definition}{Definition}
\theoremstyle{remark}

\numberwithin{equation}{section}
\title{Strategy-proof choice on strictly convex frontiers}
\author{Siwei Chen\thanks{Corresponding author. Lingnan College, Sun Yat-sen University, Guangzhou, China. Email: chensw6@mail.sysu.edu.cn. ORCID: 0000-0002-8486-2251.}
\and Pengbo Wang\thanks{Lingnan College, Sun Yat-sen University, Guangzhou, China. Email: illusory0101@gmail.com.}}
\date{September 22, 2026}

\begin{document}
\maketitle
\begin{abstract}
A collective rule must often choose a public outcome on a constrained
boundary. We characterize truthful choice with diagonal unanimity on the entire
boundary of any compact full-dimensional strictly convex body in
finite dimension at least two. On the full domain of linear support
preferences, ordinary individual strategy-proofness and diagonal
unanimity force the rule to select one fixed participant's preferred
point at every profile. The population is any fixed finite nonempty
set. No Pareto condition, boundary smoothness, or continuity or
measurability of the rule is assumed. If interior outcomes are
allowed instead, two or more participants can use positive fixed
weighted averages of preferred points for truthful non-dictatorial
compromise on the same type domain.
The proof derives regularity from the two agents' incentive
inequalities, transports a common support-measure bound across
reports, and makes every convexified attainable menu a Minkowski
summand of the feasible body. Boundary generation then forces
fixed control. The result extends the spherical public-good
location question and identifies the role of boundary-valued
choice in eliminating truthful compromise.
\end{abstract}


\section{Introduction}\label{sec:intro}

A rule that chooses a shared facility or public action should give each
participant no reason to misreport its preferred outcome. It should also
respect a location on which everyone agrees. On a line, median rules meet both
requirements while allowing several reports to affect the choice. On a closed
curved surface there is no global order that can play the same role, and an
average of preferred locations generally leaves the feasible surface. The
combination of truthfulness and a boundary constraint therefore poses a
distinct collective-choice problem.

The spherical public-good location model of \citet{cps2015,cps2016} makes this
problem especially clear. A finite society chooses one point on a sphere, and
each participant ranks outcomes by distance from an ideal point. The research
memorandum \citet{cps2015} asks whether ordinary individual
strategy-proofness together with Pareto optimality forces dictatorship. We
give an affirmative answer under the weaker requirement that the rule merely
respect agreement at unanimous profiles. The result applies to the sphere and,
more generally, to the entire boundary of every compact full-dimensional
strictly convex body.

The general formulation separates the geometry that drives the conclusion
from the special metric of a sphere. Each participant reports a unit direction
and values a public outcome by its projection in that direction. Strict
convexity gives every report a unique preferred support point on the boundary.
On a Euclidean sphere these rankings coincide with rankings by spherical or
chordal distance from the reported ideal point. On a general body they are
linear valuations of a public action under a boundary constraint; no metric
interpretation is imposed.

The boundary requirement has direct economic content. If the entire solid
body were feasible, a rule could choose a fixed weighted average of the
participants' preferred support points. That rule is truthful and unanimous,
and every participant with positive weight influences the outcome. Strict
convexity places the average in the interior whenever the positively weighted
peaks differ. The boundary removes precisely this form of truthful compromise
while leaving the reports and incentive inequalities unchanged.

Our characterization shows that the removal is complete. Every deterministic
boundary-valued rule satisfying ordinary individual strategy-proofness and
diagonal unanimity selects the preferred support point of one fixed
participant at every profile. The conclusion holds in every finite dimension
of at least two and for every fixed finite nonempty population. The boundary
may be nonsmooth, and the rule need not be continuous, measurable, anonymous,
or neutral. We assume neither Pareto optimality nor any coalitional incentive
condition.

The theorem also identifies how little efficiency is needed. Within the class
of individually strategy-proof boundary rules, agreement at unanimous
profiles, onto range, Pareto efficiency, and weak Pareto efficiency are all
equivalent to fixed dictatorship. The main result uses only the agreement
content of these conditions. It therefore isolates the force of the feasible
boundary rather than building dictatorship into a strong efficiency or
incentive requirement.

Two features make the proof nonstandard. First, incentive compatibility is
assumed for an arbitrary total rule, so continuity and measurability must be
derived rather than invoked. Second, a nonsmooth strictly convex body can have
several supporting directions at one boundary point. The support map is onto
but need not be one-to-one, and equality of preferred points cannot be used to
identify reports. The proof keeps this distinction throughout and uses normal
cones at the final geometric step.

The argument starts from the attainable menu of one participant while the
other reports are fixed. Truthfulness makes the participant's report select a
support maximizer of the convexified menu. Planar projections yield regularity
and a curvature-measure bound without assuming smoothness. A joint
support-envelope argument transports the bound across the other participant's
reports and shows that every convexified menu is a Minkowski summand of the
feasible body. The detailed measure argument also recovers the original rule
at exceptional profiles, so the conclusion is pointwise rather than merely
almost everywhere.

Strict convexity and boundary-valuedness then reduce every column menu to two
possibilities. It is either the singleton preferred by the fixed opponent or
the entire body after convexification, in which case the reporting participant
obtains its own unique support point. The second participant's incentive
constraints rule out any switch between these two forms. This fixes the
controller in the binary problem. A merging argument and successive
unilateral deviations extend the conclusion to an arbitrary finite
population without strengthening ordinary strategy-proofness.

The result provides a complete answer on this support-linear domain and a
geometric explanation for it. The characterization is pointwise, allows
singular curvature and nonsmooth support points, and makes the boundary
restriction visible in the proof itself.

\subsection{Related literature}\label{sec:literature}

The closest work is the spherical public-good model of
\citet{cps2015,cps2016}. Their established dictatorship results use strict or
coalitional strategy-proofness. The research-memorandum formulation, also
discussed by \citet{chatterjee2017}, separates those requirements from ordinary
individual strategy-proofness and poses the latter together with Pareto
optimality as the remaining question. Our theorem uses ordinary individual
strategy-proofness and replaces Pareto optimality by diagonal unanimity. It
also identifies the geometric scope of the argument by passing from a sphere
to any strictly convex frontier.

Location rules on circles and networks form a natural neighboring literature.
\citet{schummervohra2002} characterize strategy-proof location on networks,
including the restrictions created by cycles. Their setting does not by itself
reduce the present higher-dimensional problem to a circle: when reports lie in
a two-plane, the chosen outcome need not lie on the corresponding planar
section. Our proof therefore projects outcomes into a convex planar body and
analyzes the resulting body-valued rule. For rules whose continuity is assumed,
the topological argument of \citet[Theorem~13, pp.~97--98]{chichilnisky1983}
provides a related benchmark. Appendix~\ref{app:continuous} gives the
finite-dimensional version used for comparison; the main theorem instead
derives the regularity it needs from incentives.

Generalized median and issue-voting results explain how non-dictatorial
strategy-proof choice can survive on ordered or median structures.
\citet{BarberaMassoNeme1997,BarberaMassoNeme1999} study constrained finite
alternative sets and domains that preserve generalized median voter schemes.
\citet{NehringPuppe2007} characterize strategy-proof social choice on finite
median spaces through consistent voting by issues under a rich single-peaked
domain. Those constructions use discrete betweenness or issue structure. A
strictly convex frontier is a continuous outcome space without a global median
operation, and the support-linear domain permits weak orders when distinct
directions expose the same point.

The compact-range analysis of \citet{BarberaMassoSerizawa1998} is closer in
allowing multidimensional outcomes, but it works with a rich domain of
multidimensional single-peaked preferences on compact ranges. Our domain is
considerably more specific: reports are support-linear valuations, while the
range is the entire boundary and hence has empty ambient interior. This
combination makes the menu geometry decisive. If the range is enlarged to the
solid body, the weighted-average rule in Proposition~\ref{prop:body-average}
immediately restores truthful non-dictatorial compromise.

Other dictatorship results emphasize richness of preferences rather than the
geometry of a thin range. \citet{Zhou1991} studies public-good environments with
continuous quasiconcave preferences and, in one formulation, an abundance
condition that includes a broad class of quadratic preferences.
\citet{AchuthankuttyRoy2018} study finite domains of strict orders: their
top-circular structure yields dictatorship only with additional conflict
conditions, and top-circularity alone permits a non-dictatorial example. The
support-linear weak-order domain here satisfies neither type of richness
assumption. Our characterization instead follows from the joint restrictions
that ordinary incentive compatibility places on all attainable menus along a
strictly convex boundary.

\section{Model and result}\label{sec:model}

Let \(N=\{1,\ldots,n\}\), with \(1\leq n<\infty\).
The group selects a common public outcome from \(\partial K\),
where \(K\subset\mathbb R^d\) is compact, full-dimensional and
strictly convex, and \(2\leq d<\infty\). Here \(\partial K\) and
\(\operatorname{int}K\) denote the boundary and interior of \(K\).
Full-dimensionality means \(\operatorname{int}K\ne\varnothing\),
and strict convexity means that
\((1-t)y+tz\in\operatorname{int}K\) for all distinct \(y,z\in K\)
and every \(t\in(0,1)\). Equivalently, \(\partial K\) contains no
nondegenerate line segment. There are no transfers.

The type space is the Euclidean unit sphere
\[
 S^{d-1}=\{p\in\mathbb R^d:\|p\|=1\}.
\]
A type \(p\) induces the preference relation
\[
 y\succeq_p z \quad\Longleftrightarrow\quad p\cdot y\geq p\cdot z.
\]
Scaling a nonzero linear valuation by a positive constant
does not change this ranking, so unit normalization imposes
no cardinal comparison between participants.

For a nonzero vector \(P\), let
\[
 h_K(P)=\max_{y\in K}P\cdot y,\qquad
 \{X(P)\}=\operatorname*{argmax}_{y\in K}P\cdot y,\qquad
 x_K(p)=X(p)\quad(p\in S^{d-1}).
\]
Strict convexity makes the support point unique, and it lies
on \(\partial K\). The support map is onto the boundary
but need not be one-to-one. Distinct reported directions
can have the same peak while ranking other outcomes
differently. Reports therefore specify valuations, not
merely a top-point selector. The origin need not lie in \(K\).

When \(K\) is a Euclidean unit ball, \(x_K(p)=p\).
Great-circle distance \(\arccos(p\cdot y)\) and chordal
distance induce the same preferences as the inner product.
For a general body we maintain the linear valuations;
we do not identify them with distance preferences from
\(x_K(p)\).

For a profile \(p=(p_1,\ldots,p_n)\), write \(p_{-i}\) for the
reports of all participants other than \(i\).

\begin{definition}\label{def:sp}
A deterministic total rule is a function
\(F:(S^{d-1})^n\to\partial K\) defined at every profile.
It is \emph{ordinarily individually strategy-proof}, or IC, if
\begin{equation}\label{eq:ic}
 p_i\cdot F(p_i,p_{-i})\geq p_i\cdot F(r_i,p_{-i})
\end{equation}
for every \(i\), true type \(p_i\), alternative report \(r_i\)
and opponents' profile \(p_{-i}\).
It is \emph{diagonally unanimous} if
\begin{equation}\label{eq:unanimity}
                      F(p,\ldots,p)=x_K(p)
                      \qquad(p\in S^{d-1}).
\end{equation}
A \emph{fixed dictatorship} selects \(x_K(p_j)\) at every
profile for a single fixed label \(j\).
\end{definition}

The incentive inequalities are weak and concern arbitrary
unilateral deviations, evaluated at the same true type
on both sides. Diagonal unanimity constrains only profiles
with identical reported directions. We assume no
continuity, measurability, symmetry or coalitional
incentive property.

We can now state the characterization.

\begin{theorem}[Strategy-proof choice on strictly convex frontiers]\label{thm:frontier}
Let \(2\leq d<\infty\), \(1\leq n<\infty\), and let
\(K\subset\mathbb R^d\) be compact, full-dimensional
and strictly convex. A deterministic total rule
\(F:(S^{d-1})^n\to\partial K\) is ordinarily individually
strategy-proof and diagonally unanimous if and only if
there is a fixed \(j\in\{1,\ldots,n\}\) such that
\[
                      F(p_1,\ldots,p_n)=x_K(p_j)
                      \qquad\text{for every profile}.
\]
No boundary smoothness or rule continuity or
measurability is required.
\end{theorem}

The conclusion concerns the original rule at every
profile, not merely a regular representative almost
everywhere. It includes coincident and opposite
reports and nonsmooth support points with multiple
normals. The label may depend on the rule and on
\(K\), but never on the reports. For a spherical
location problem, the theorem says that the same
person's ideal point must always be chosen.

The proof below first establishes the two-agent
classification and then lifts it to every fixed finite
population. The population argument uses successive
individual incentive comparisons with a common evaluator;
ordinary IC is not strengthened to group-strategy-proofness.
Appendix~\ref{app:population} proves the efficiency
equivalences and records the one-dimensional and continuous
benchmarks.

For clarity, call an outcome \emph{Pareto efficient}
if no other boundary point weakly raises all
reported utilities and strictly raises one.
It is \emph{weakly Pareto efficient} if no boundary
point strictly raises every reported utility. The rule is
\emph{onto} if its range is all of \(\partial K\).

\begin{corollary}\label{cor:efficiency}
Within the class of individually IC total
boundary-valued rules, fixed dictatorship,
diagonal unanimity, onto range, Pareto efficiency,
and weak Pareto efficiency are equivalent.
Weak efficiency required only at diagonal
profiles already implies these properties.
\end{corollary}

At a diagonal profile the common support point
strictly dominates every other feasible outcome.
The theorem thus uses only the agreement content
of efficiency. Ontoness yields the same agreement
under IC by a sequence of unilateral comparisons.
The complete implications are proved in
Appendix~\ref{app:population}.

The diagonal condition cannot simply be dropped:
a constant boundary rule is IC. The dimensional
restriction at one is also sharp. On the two
endpoints of an interval, the two-person rule
that chooses the upper endpoint whenever either
report prefers it is truthful and unanimous,
but neither fixed label controls both
disagreement profiles. Appendix~\ref{app:population}
verifies this comparison. Strict convexity,
in turn, supplies unique support points and
the menu dichotomy used in the proof.

The role of the boundary can be isolated by changing only the feasible outcome
space from \(\partial K\) to \(K\).

\begin{proposition}\label{prop:body-average}
Let \(\lambda_i\geq0\), \(\sum_i\lambda_i=1\).
The body-valued rule
\[
                      G(p)=\sum_{i=1}^n\lambda_i x_K(p_i)
\]
is continuous, individually IC and diagonally unanimous.
If at least two weights are positive, it is not a
fixed dictatorship. Whenever the positively weighted
peaks are not all identical, \(G(p)\) lies in
\(\operatorname{int}K\).
\end{proposition}
\begin{proof}
Continuity follows from continuity of the support map,
proved in Appendix~\ref{app:geometry}. Under a unilateral
deviation, the truthful-minus-deviating utility difference is
\[
            \lambda_i p_i\cdot
            \bigl[x_K(p_i)-x_K(r_i)\bigr]\geq0.
\]
On the diagonal, all summands are the same support point.
A nontrivial convex combination of distinct points of a
strictly convex body lies in its interior. With two
positive weights, choose their support points distinct;
the resulting interior point cannot equal any agent's
boundary peak. The rule is therefore non-dictatorial.
\end{proof}

This is an average of physical outcomes, not an assumption
about richer reports or additional preference rankings.
It can use everyone's report with positive weight.
The comparison shows why the restriction to a single
boundary point matters: it rules out interior
compromise even when that compromise is truthful.

\section{Proof}\label{sec:menus}

This section gives the logical structure of the classification proof.
The appendices supply the complete arguments, including the measure
and exceptional-set steps used to reach every profile of the original
rule. We first analyze a binary rule, meaning a rule for two participants.
At opponents' reports \(p_{-i}\), agent \(i\)'s attainable menu is
\[
 M_i(p_{-i})=\{F(r_i,p_{-i}):r_i\in S^{d-1}\}.
\]
IC makes truthful choice a support maximizer
of that menu. Convexifying the menu does not
change any linear support value. The proof
uses the two agents' incentive restrictions
jointly to determine which convexified menus
can coexist.

For sets \(A,B\subset\mathbb R^d\), their Minkowski sum is
\(A+B=\{a+b:a\in A,\ b\in B\}\). A compact convex set \(C\)
is a Minkowski summand of \(K\) if \(K=C+D\) for some compact
convex set \(D\).

\subsection{Menus and planar projections}

\begin{lemma}\label{lem:exposure}
The support function \(h_K\) is \(C^1\) off the
origin, with \(\nabla h_K(P)=X(P)\).
The map \(x_K\) is continuous and onto
\(\partial K\). For every \(\eta>0\), the
support loss \(h_K(p)-p\cdot y\) has a
strictly positive lower bound over unit
\(p\) and \(y\in K\) with
\(\|y-x_K(p)\|\geq\eta\), whenever this set
is nonempty. Every two-dimensional
orthogonal projection of \(K\) is a
full-dimensional strictly convex body
in its plane.
\end{lemma}

This exposure property uses compactness and
uniqueness of the supporting point, rather
than a positive curvature bound.
For a body-valued binary IC rule with diagonal
\(x_K\), reporting the opponent's type gives
\[
 0\leq h_K(p)-p\cdot F(p,q)
       \leq h_K(p)-p\cdot x_K(q).
\]
Thus \(F(p,q)\) approaches \(x_K(p)\)
uniformly as \(q\) approaches \(p\).

\begin{lemma}\label{lem:ambient-menus}
For a boundary-valued binary IC rule,
extend reports homogeneously by
\(F(P,Q)=F(P/\|P\|,Q/\|Q\|)\) for nonzero \(P,Q\).
At any fixed \(Q\neq0\), set
\[
 M_Q=\{F(P,Q):P\neq0\},\qquad
 C_Q=\operatorname{conv}M_Q.
\]
Both sets are compact, truthful choice
maximizes \(P\cdot y\) over \(C_Q\), and
\[
 M_Q=C_Q\cap\partial K,\qquad
 C_Q=\operatorname{conv}(C_Q\cap\partial K),\qquad
 \operatorname{ext}C_Q\subset M_Q.
\]
With diagonal unanimity, \(X(Q)\in M_Q\).
The row analogues hold as well.
\end{lemma}

Appendix~\ref{app:geometry} proves these facts.
For the intermediate body-valued argument
we instead define a menu as the closed
convex hull of the actual image. It is
compact inside \(K\); closedness of
an interior-valued actual image is
not assumed.

If both reports lie in a two-plane, their
utilities depend only on the orthogonal
projection of the outcome. This produces a
rule valued in the projected body, even
though the original outcome need not lie
in that plane. The appropriate planar
statement is therefore body-valued.

\begin{lemma}\label{lem:planar-body}
Let \(L\subset\mathbb R^2\) be compact,
full-dimensional and strictly convex.
Suppose \(g:(S^1)^2\to L\) is individually
IC in both coordinates and
\(g(p,p)=x_L(p)\). Then \(g\) is jointly
continuous, and every closed convex
column or row menu is a Minkowski
summand of \(L\).
Writing \(t(a)=(-\sin a,\cos a)\) and
\(\mu_L=D^2h_L+h_L\,da\), every angular
section satisfies
\[
                         Dg=t\,\nu,\qquad
                         0\leq\nu\leq\mu_L
\]
as vector measures.
The curvature measure \(\mu_L\) is atomless,
but can have a singular-continuous part.
For the unit disk, the section is
absolutely continuous with derivative
\(\lambda(a)t(a)\), \(0\leq\lambda\leq1\)
almost everywhere.
\end{lemma}

Appendix~\ref{app:plane} starts with the
incentive inequalities themselves.
Outcome increments lie in narrow
tangent cones. Two independent scalar
projections give local integrability
without a prior measurability assumption.
The two support envelopes then force
the mixed weak derivative to vanish
on each nonparallel angle chamber.

The resulting decomposition has two
one-variable components with positive
tangent derivative measures.
To connect them at the diagonal,
the proof uses legal offset slices
and an \(L^1\) limit; an almost-everywhere
identity cannot simply be restricted
to the diagonal. The two measures
sum to the body's curvature measure.
Monotone sandwiches recover the
original point values, and an
antipodal trace argument closes
the remaining profiles. The proof
keeps singular-continuous curvature
rather than discarding it.

\subsection{Global menu decomposition}

The planar projections still have to
be compatible across columns. For this intermediate step,
we relax only the outcome requirement and consider an auxiliary
body-valued binary IC rule
\(\widetilde F:(S^{d-1})^2\to K\) with
\(\widetilde F(p,p)=x_K(p)\). This notation distinguishes the temporary
body-valued rule from the boundary-valued rule in the theorem. On
\(\Omega=(\mathbb R^d\setminus\{0\})^2\),
define its column and row support
envelopes by
\[
 U(P,Q)=P\cdot \widetilde F(P,Q),\qquad
 V(P,Q)=Q\cdot \widetilde F(P,Q).
\]

At differentiability points, the gradients of the scalar envelopes recover the
chosen outcome. Convexity in the corresponding report makes their Hessians
positive. The next lemma retains these facts as distributional statements when
the auxiliary rule is nonsmooth. Informally, the matrix measures \(A\) and
\(B\) record the monotone movement of the outcome generated by each report.
Their transport identities say that the first report's movement measure is
preserved along the relevant changes in the second report. Contraction toward
the diagonal then compares that measure with \(D^2h_K\); positivity prevents
the other coordinate from cancelling any mass. Integrating the resulting
order shows that \(h_K-h_{C_q}\) is a support function, and continuity extends
the conclusion from regular slices to every actual column. This is the route
from the distributional identities to the Minkowski decomposition below.

\begin{lemma}\label{lem:transport}
The envelopes \(U,V\) are jointly
continuous. The auxiliary rule is
measurable for completed Lebesgue
measure on \(\Omega\), and
\[
 A=D_P\widetilde F=D^2_{PP}U,\qquad
 B=D_Q\widetilde F=D^2_{QQ}V
\]
are symmetric positive semidefinite
matrix Radon measures.
On linearly independent pairs,
\(A\) is invariant under
\((P,Q)\mapsto(P,Q+tP)\) and
\((P,Q)\mapsto(P,\epsilon Q)\),
\(\epsilon>0\). In particular, for
\(\Phi_\epsilon(P,Q)=(P,P+\epsilon Q)\),
the componentwise scalar-distribution
pullback satisfies \(\Phi_\epsilon^*A=A\).
\end{lemma}

The joint mixed-derivative symmetries
and homogeneity generate this
transport. Appendix~\ref{app:transport}
specifies its distributional and
Jacobian conventions. Contracting
towards the diagonal yields
\(\widetilde F\circ\Phi_\epsilon\to X\)
locally uniformly and
\[
 D_P(\widetilde F\circ\Phi_\epsilon)
                 =A+\Phi_\epsilon^*B.
\]
Both terms are positive. The diagonal
therefore bounds each menu's measure
by the same environmental support
measure:
\[
                 0\preceq A\preceq
                 (D^2h_K)\otimes\mathcal L_Q^d.
\]

\begin{proposition}\label{prop:minkowski}
For the body-valued binary IC rule
with diagonal \(x_K\), every fixed
opponent \(q\) has a compact convex
complement \(D_q\) such that
\[
 K=C_q+D_q,\qquad
 C_q=\overline{\operatorname{conv}}
       \{\widetilde F(p,q):p\in S^{d-1}\}.
\]
The same conclusion holds for every row.
\end{proposition}

For \(d\geq3\), a cutoff around parallel
report pairs removes possible residual
mass on that set. In dimension two,
the planar lemma applies directly.
Joint continuity of the envelopes
then upgrades the distributional
bound to each actual column:
\(h_K-h_{C_q}\) is a support function.
Appendix~\ref{app:budget} provides
the full argument, including singular
matrix measures and the exceptional set.

The decomposition is a joint incentive
restriction. The choice available to
one agent is a summand of the same
feasible body that bounds the other
agent's options. This connection across
reports, together with the fact that actual
menu points lie on \(\partial K\), determines control.

\subsection{Fixed control}

We now return to the boundary-valued rule. The preceding proposition and the
menu identities in Lemma~\ref{lem:ambient-menus} turn the global decomposition
into a pointwise classification.

\begin{proof}[Proof of the binary case]
Fix a nonzero second report \(Q\). Under the homogeneous extension,
\(C_Q=C_{Q/\|Q\|}\). Applying Proposition~\ref{prop:minkowski} to
\(q=Q/\|Q\|\) therefore gives compact convex sets \(C_Q,D_Q\) such that
\[
                         K=C_Q+D_Q.
\]
For a compact convex set \(C\) and a nonzero vector \(u\), write
\(E_C(u)=\arg\max_{y\in C}u\cdot y\) for its exposed face. Support
maximizers add under Minkowski addition, so
\[
                         E_K(u)=E_{C_Q}(u)+E_{D_Q}(u).
\]
Strict convexity makes \(E_K(u)\) a singleton. A sum of two nonempty sets can
be a singleton only if each summand is a singleton. Hence every exposed face
of \(C_Q\) is a singleton.

This property already rules out a nonsingleton lower-dimensional menu. If
\(C_Q\) were lower-dimensional and contained two distinct points, choose a
nonzero vector \(u\) perpendicular to
\(\operatorname{lin}(C_Q-C_Q)\). The functional \(u\cdot y\) is then constant
on \(C_Q\), so \(E_{C_Q}(u)=C_Q\), contradicting the singleton-face property.
Thus a lower-dimensional \(C_Q\) must be a singleton.

Suppose instead that \(C_Q\) is full-dimensional. Since
\(M_Q\subset\partial K\subset K\) and \(K\) is convex,
\(C_Q=\operatorname{conv}M_Q\subset K\). If the inclusion is proper, choose
\(a\in\operatorname{int}C_Q\), which also belongs to
\(\operatorname{int}K\). There is a point
\(b\in\operatorname{int}K\setminus C_Q\). To see this explicitly, take
\(z\in K\setminus C_Q\) and a ball \(B(y_0,r)\subset K\). For \(t>0\), the
point \(z_t=(1-t)z+ty_0\) satisfies
\(B(z_t,tr)\subset K\), and hence is interior; because \(C_Q\) is closed and
\(z\) has positive distance from it, sufficiently small \(t\) also gives
\(z_t\notin C_Q\). We may take \(b=z_t\).

The set
\[
 \{s\in[0,1]:(1-s)a+sb\in C_Q\}
\]
is a closed interval \([0,s_*]\) with \(0<s_*<1\). Its endpoint
\(y=(1-s_*)a+s_*b\) lies on \(\partial C_Q\). It also lies in
\(\operatorname{int}K\): choose \(\epsilon>0\) such that
\(B(a,\epsilon),B(b,\epsilon)\subset K\). Convexity gives
\[
 B((1-s)a+sb,\epsilon)
   =(1-s)B(a,\epsilon)+sB(b,\epsilon)\subset K
\]
for every \(s\in[0,1]\). A nonzero supporting normal
of \(C_Q\) exists at \(y\). Its exposed face is a singleton, so it exposes
\(y\) alone. If \(y\) were a nontrivial convex combination of two points of
\(C_Q\), both points would attain the same support value and would lie in that
singleton face. Thus \(y\in\operatorname{ext}C_Q\).

Lemma~\ref{lem:ambient-menus} gives
\(\operatorname{ext}C_Q\subset M_Q\subset\partial K\). This contradicts
\(y\in\operatorname{int}K\). We have proved that \(C_Q\) is either a singleton
or all of \(K\). In the singleton case, diagonal unanimity gives
\(X(Q)\in M_Q\subset C_Q\), so \(C_Q=\{X(Q)\}\) and every actual value in the
column equals \(X(Q)\). In the other case, \(C_Q=K\). First-coordinate IC says
that \(F(P,Q)\) maximizes \(P\cdot y\) over the actual menu and hence over its
convex hull \(K\). The maximizer over \(K\) is unique, so
\(F(P,Q)=X(P)\) for every \(P\), including reports in any exceptional set used
in the preceding measure argument.

Restrict again to unit reports. Let \(\mathcal S\) be the set of opponent
reports \(r\) whose columns are constant, and let \(\mathcal I\) be the set of
opponent reports \(q\) whose columns select the first participant's peak. The
two sets cover \(S^{d-1}\). They are disjoint because a column of both forms
would make \(X(p)\) constant in \(p\), whereas the support map is onto the
nonsingleton boundary of a full-dimensional body. Suppose both are nonempty. For
\(q\in\mathcal I\), \(r\in\mathcal S\), and every first report \(p\), we have
\[
                         F(p,q)=X(p),\qquad F(p,r)=X(r).
\]
When the second participant's true type is \(q\), IC for the deviation from
\(q\) to \(r\) gives
\[
                         q\cdot X(p)\geq q\cdot X(r)
                         \qquad\text{for every }p.
\]
Set \(p=-q\). The point \(X(-q)\) uniquely minimizes \(q\cdot y\) over \(K\),
so the displayed inequality can hold only if \(X(r)=X(-q)\). Notice that the
inequality is weak and that no conclusion \(r=-q\) is used.

Fix \(q_0\in\mathcal I\) and \(r_0\in\mathcal S\), and put
\(c=X(r_0)=X(-q_0)\). Pairing any \(r\in\mathcal S\) with \(q_0\) gives
\(X(r)=c\); pairing any \(q\in\mathcal I\) with \(r_0\) gives \(X(-q)=c\).
Consequently
\[
 \mathcal S\subset N_K(c)\cap S^{d-1},\qquad
 \mathcal I\subset -N_K(c)\cap S^{d-1},
\]
where
\[
 N_K(c)=\{u\in\mathbb R^d:u\cdot(y-c)\leq0
                       \text{ for every }y\in K\}.
\]
This step uses only the definition of a normal cone and remains valid when
several normals support \(K\) at \(c\).

Choose a closed ball \(\overline B(y_0,\rho)\subset K\), with \(\rho>0\), and
set \(a=c-y_0\). For any nonzero \(u\in N_K(c)\), insert
\(y=y_0+\rho u/\|u\|\) into the defining normal inequality. It gives
\[
                         a\cdot u\geq\rho\|u\|>0.
\]
In particular \(a\ne0\), since \(r_0\in N_K(c)\) is nonzero. Since
\(d\geq2\), choose a unit vector
\(p_*\perp a\). Neither \(p_*\) nor \(-p_*\) can belong to \(N_K(c)\), because
the displayed inequality would give \(0\geq\rho\). Hence
\(p_*\notin N_K(c)\cup[-N_K(c)]\), contradicting the fact that
\(\mathcal S\cup\mathcal I=S^{d-1}\). One of the two classes must therefore be
empty. The binary rule selects \(X(p)\) at every profile or selects \(X(q)\)
at every profile, with a fixed controlling label.
\end{proof}

\subsection{Finite populations}

The binary conclusion extends by induction. The argument below spells out the
unilateral comparisons because this is where a simultaneous deviation must be
avoided.

\begin{proof}[Completion of the proof of Theorem~\ref{thm:frontier}]
For any type \(p\), the point \(X(p)\) is its unique best outcome. The same
point is the unique worst outcome for type \(-p\): if \(y\ne X(p)\), strict
exposure gives \(p\cdot y<p\cdot X(p)\), and therefore
\((-p)\cdot y>(-p)\cdot X(p)\). This statement does not require
\(X(-p)=-X(p)\), which need not hold for a general body.

The one-person case follows immediately from diagonal unanimity, since every
profile is diagonal. The binary case was just proved. Suppose the result holds
for \(n-1\) participants, where \(n\geq3\). Write
\(R=(p_1,\ldots,p_{n-2})\) for the first \(n-2\) reports and merge the last two
reports by defining
\[
                         G(R,q)=F(R,q,q).
\]
The first \(n-2\) coordinates inherit IC directly. To verify IC for the merged
coordinate, take its true type to be \(q\) and its alternative report to be
\(q'\). Ordinary IC applied first to participant \(n-1\) and then to
participant \(n\) gives
\[
 q\cdot F(R,q,q)
 \geq q\cdot F(R,q',q)
 \geq q\cdot F(R,q',q').
\]
In the first comparison, participant \(n\)'s report stays at \(q\). In the
second, participant \(n-1\)'s report stays at \(q'\), while participant \(n\),
whose true type is still \(q\), deviates to \(q'\). The evaluator is \(q\) in
both inequalities, so they telescope. No coalitional incentive condition is
used. The rule \(G\) is also diagonally unanimous, and the induction hypothesis
gives it one fixed controlling coordinate.

Suppose first that an unmerged participant \(i\leq n-2\) controls \(G\). Fix
an arbitrary \(R\), let \(t=-p_i\), and set \(c=X(p_i)\). At the last-two
profile \((t,t)\),
\[
                         F(R,t,t)=G(R,t)=c.
\]
The point \(c\) is the unique worst outcome for type \(t\). Let participant
\(n-1\) have true type \(t\) and deviate to an arbitrary report \(a\), while
participant \(n\)'s report remains \(t\). IC gives
\[
                         t\cdot c\geq t\cdot F(R,a,t).
\]
Every outcome in \(K\) gives type \(t\) utility at least \(t\cdot c\), with
equality only at \(c\). Thus \(F(R,a,t)=c\). Now keep \(a\) fixed and let
participant \(n\), whose report \(t\) is still truthful, deviate to an
arbitrary \(b\). A second unilateral comparison gives
\[
                         t\cdot F(R,a,t)=t\cdot c
                         \geq t\cdot F(R,a,b),
\]
and uniqueness of the worst point again yields \(F(R,a,b)=c\). Since
\(R,a,b\) were arbitrary, the original participant \(i\) controls \(F\) at
every profile.

It remains to consider the case in which the merged coordinate controls
\(G\). For every fixed outside profile \(R\), define the binary section
\[
                         H_R(a,b)=F(R,a,b).
\]
It is IC in both displayed coordinates, and
\(H_R(q,q)=G(R,q)=X(q)\). The binary theorem therefore makes \(H_R\) either
the first-coordinate peak rule or the second-coordinate peak rule.

The controlling coordinate cannot change with \(R\). To prove this, take
outside profiles \(R,R'\) that differ only in one participant's report, from
\(r\) to \(r'\). Suppose \(H_R\) selects \(X(a)\) while \(H_{R'}\) selects
\(X(b)\). IC for that outside participant, evaluated at true type \(r\), says
\[
                         r\cdot X(a)\geq r\cdot X(b)
                         \qquad\text{for every }a,b\in S^{d-1}.
\]
Choosing \(a=-r\) and \(b=r\) makes the left side the minimum and the right
side the maximum of \(r\cdot y\) on \(K\). These values differ because a
full-dimensional body has positive width in every nonzero direction. If the
two sections have the opposite controllers, the IC
inequality is reversed, and choosing \(a=r\), \(b=-r\) gives the same
contradiction. Any two outside profiles can be connected by finitely many
single-coordinate changes. Hence every section has the same controller, and
one of the original labels \(n-1,n\) controls \(F\) at every profile. This
completes the induction.

Conversely, consider a fixed peak rule with controlling label \(j\). If
\(i=j\), then
\[
 p_j\cdot x_K(p_j)\geq p_j\cdot x_K(r_j)
\]
for every alternative report \(r_j\), because \(x_K(p_j)\) maximizes the true
linear valuation on \(K\). If \(i\ne j\), changing participant \(i\)'s report
leaves the outcome unchanged, so the IC inequality holds with equality. The
rule is diagonally unanimous as well. This proves both directions of the
theorem.
\end{proof}

\section{Conclusion}

On a strictly convex frontier, ordinary individual strategy-proofness and
agreement determine the rule completely: one fixed participant's preferred
support point is selected at every profile. This conclusion holds for every
finite population in dimension at least two, without continuity,
measurability, smoothness, anonymity, neutrality, or a Pareto assumption.

The proof derives its own regularity from incentives. It converts attainable
menus into Minkowski summands, uses the boundary to reduce them to two possible
forms, and then fixes the controlling label through ordinary unilateral
comparisons. Normal cones keep the conclusion valid at nonsmooth support
points, and the finite-population argument uses only individual deviations.
The body-valued comparison pinpoints the role of frontier feasibility:
allowing interior outcomes restores truthful shared influence through fixed
weighted averages.

\clearpage
\phantomsection

\clearpage
\appendix
\phantomsection
\addcontentsline{toc}{part}{Appendices}
\begin{center}
 {\Large\bfseries Appendices}\par\medskip
\end{center}
\section{Geometry and menus}\label{app:geometry}

For a nonempty compact convex set $K\subset\R^d$ write
$h_K(P)=\max_{y\in K}P\cdot y$ and
$E_K(P)=\{y\in K:P\cdot y=h_K(P)\}$ for $P\ne0$.
Support functions add under Minkowski sums and satisfy
$E_{C+D}(P)=E_C(P)+E_D(P)$; both facts follow by maximizing
the same linear functional separately in the two summands,
the second because two nonnegative support gaps sum to zero
only if each vanishes.
Equal support functions force equal sets: if $x\in A\setminus B$
with $B$ compact convex and nonempty, let $y$ be the point of
$B$ nearest $x$; convexity gives $(x-y)\cdot(z-y)\leq0$ for all
$z\in B$, so $u=x-y\ne0$ satisfies
$u\cdot x>\max_{z\in B}u\cdot z$, contradicting $h_A=h_B$.
Only this nearest-point argument is used, not a general
separation theorem.

For compact full-dimensional convex $K$, strict convexity is
equivalent to $E_K(P)$ being a singleton for every $P\ne0$.
Indeed, a maximizer of a nonzero linear functional cannot be an
interior point, so $E_K(P)$ is a convex subset of $\partial K$,
and two distinct points of it would span a segment in $\partial K$.
Conversely, if a nondegenerate segment $[a,b]$ lies in $\partial K$,
let $P\ne0$ support $K$ at its midpoint. Then
$P\cdot a,P\cdot b\leq h_K(P)$ and their average equals $h_K(P)$,
so both $a$ and $b$ belong to $E_K(P)$.

\begin{proof}[Proof of Lemma~\ref{lem:exposure}]
By this equivalence, $E_K(P)$ is a singleton, denoted $X(P)$.
If $P_j\to P\ne0$, compactness gives a convergent subsequence
of any proposed nonconvergent sequence $X(P_j)$; its limit
maximizes $P$ and so equals $X(P)$. Thus $X$ is continuous.
For small $v$,
\[
 v\cdot X(P)\leq h_K(P+v)-h_K(P)
              \leq v\cdot X(P+v).
\]
Continuity proves differentiability with $\nabla h_K(P)=X(P)$.
Every boundary point $y$ has a supporting nonzero normal:
take $z_k\notin K$ with $z_k\to y$, let $y_k$ be the point of
$K$ nearest $z_k$ and $r_k$ the unit vector along $z_k-y_k$.
Convexity gives $r_k\cdot(x-y_k)\leq0$ for all $x\in K$, and
$\|y_k-y\|\leq2\|z_k-y\|\to0$, so any subsequential limit $r$
of the $r_k$ supports $K$ at $y$. Strict convexity makes $y$
the unique maximizer of $r$. Therefore
$x_K=X|_{S^{d-1}}$ is onto $\partial K$, but need not be
injective: two arcs of $\partial K$ meeting at a corner give
that corner two distinct supporting unit normals. Note also
that $X(P)\in\partial K$ and $X(\lambda P)=X(P)$ for
$\lambda>0$, both used freely below.

For $\eta>0$, the set of pairs
$(p,y)\in S^{d-1}\times K$ with $\|y-x_K(p)\|\geq\eta$
is compact. When nonempty it contains no support maximizer,
so the minimum
\[
 \gamma_K(\eta)=\min\{h_K(p)-p\cdot y:
             p\in S^{d-1},\,y\in K,\,\|y-x_K(p)\|\geq\eta\}
\]
is strictly positive. For an empty feasible set we define it
as $+\infty$. Compactness and continuity also give
\[
 \zeta_K(\delta)=\max_{\substack{p,q\in S^{d-1}\\\|p-q\|\leq\delta}}
      [h_K(p)-p\cdot x_K(q)]\longrightarrow0.
\]
Let $\omega_K(\delta)$ denote the supremum of $\|y-x_K(p)\|$
over pairs with $h_K(p)-p\cdot y\leq\zeta_K(\delta)$; it is
bounded by the diameter of $K$ and nondecreasing in $\delta$.
It tends to zero, and this is the step on which the later
appendices rely, so we give it in full. Fix $\eta>0$.
Since $\gamma_K(\eta)>0$ and $\zeta_K(\delta)\to0$, choose
$\delta_0>0$ with $\zeta_K(\delta)<\gamma_K(\eta)$ for
$\delta<\delta_0$; when the constraint set defining
$\gamma_K(\eta)$ is empty this holds for every $\delta$ by
the convention $\gamma_K(\eta)=+\infty$. If a pair $(p,y)$
had $h_K(p)-p\cdot y\leq\zeta_K(\delta)$ and
$\|y-x_K(p)\|\geq\eta$, it would be feasible in the
minimization defining $\gamma_K(\eta)$ and would give
$\gamma_K(\eta)\leq\zeta_K(\delta)<\gamma_K(\eta)$.
Hence no such pair exists and $\omega_K(\delta)\leq\eta$
for $\delta<\delta_0$. As $\eta$ was arbitrary,
$\omega_K(\delta)\to0$.
IC and the option to report the opponent's type give
\[
 0\leq h_K(p)-p\cdot F(p,q)
 \leq h_K(p)-p\cdot x_K(q),
\]
hence $\|F(p,q)-x_K(p)\|\leq\omega_K(\|p-q\|)$.
No quadratic modulus or curvature bound is used.

If $L$ is a two-dimensional linear subspace, the projection
$K_L=\Pi_LK$ has interior in $L$, and for $u\in L\setminus\{0\}$,
$E_{K_L}(u)=\Pi_LE_K(u)$ is a singleton. By the equivalence
above, applied in the plane $L$, $K_L$ is strictly convex and
the same conclusions apply to it.
\end{proof}

\begin{proof}[Proof of Lemma~\ref{lem:ambient-menus}]
Fix a binary opponent $Q\ne0$ and let
$M_Q=\{F(P,Q):P\ne0\}$. IC says that $F(P,Q)$ maximizes
$P\cdot y$ over this actual menu.
Suppose $y_j\in M_Q$ converges to some $y$. Since
$M_Q\subset\partial K$ and $\partial K$ is closed, the limit
$y$ lies in $\partial K$; take any
supporting normal $r$ of $K$ at $y$. IC gives
$r\cdot F(r,Q)\geq r\cdot y_j\to h_K(r)$.
Uniqueness of $E_K(r)$ forces $F(r,Q)=y$. Thus $M_Q$
is closed and compact. Its convex hull $C_Q$ is compact
by the finite-dimensional convex-hull theorem.
For $y\in C_Q\cap\partial K$, represent $y$ as a finite
convex combination of menu points. A support normal
of $K$ at $y$ forces every positive-weight point of
that combination to equal $y$, so $y\in M_Q$. Hence
\[
 M_Q=C_Q\cap\partial K,\qquad
 C_Q=\operatorname{conv}(C_Q\cap\partial K),\qquad
 \operatorname{ext}C_Q\subset M_Q .
\]
For the last inclusion, let $e\in\operatorname{ext}C_Q$ and use
finite-dimensionality to write
$e=\sum_{j=1}^k\lambda_jm_j$, where $m_j\in M_Q$,
$\lambda_j\geq0$ and $\sum_j\lambda_j=1$. Choose an index
$j_0$ with $\lambda_{j_0}>0$. If $\lambda_{j_0}=1$, then
$e=m_{j_0}\in M_Q$. Otherwise put
\[
 y=\frac{1}{1-\lambda_{j_0}}
       \sum_{j\ne j_0}\lambda_jm_j\in C_Q.
\]
Then $e=\lambda_{j_0}m_{j_0}+(1-\lambda_{j_0})y$ is a strict
convex combination of two points of $C_Q$. Extremality forces
$m_{j_0}=y=e$, and hence again $e\in M_Q$.
Everything above uses only the first agent's IC; diagonal
unanimity enters solely to give $X(Q)\in M_Q$. The row
statements follow by applying the same argument to
$\widetilde F(P,Q)=F(Q,P)$, whose two IC conditions and
diagonal values are those of $F$ with the labels exchanged;
no symmetry of $F$ is required.
\end{proof}
\section{Planar incentives}\label{app:plane}

Let $L\subset\R^2$ be compact, full-dimensional and strictly convex.
Write
\[
 n(a)=(\cos a,\sin a),\quad t(a)=n'(a),\quad
 h(a)=h_L(n(a)),\quad z(a)=x_L(n(a)).
\]
Angles are lifted periodically to $\R$; measures such as
$\mu_L$ below are finite on the circle and, read on the lift,
periodic and locally finite rather than of finite total mass
on $\R$. Only the latter is used. The rule
$G:\Circ\times\Circ\to L$ is body-valued, is IC in both coordinates
for the corresponding linear utilities, and has $G(a,a)=z(a)$.
At the outset define the convexified column menu by
\[
 C_b=\overline{\operatorname{conv}}\{G(a,b):a\in\R\}.
\]
It is compact as a closed subset of $L$. We have not yet asserted
that the unconvexified image of a column is closed or continuous.
Its support function is
$c_b(a)=n(a)\cdot G(a,b)$, because truthful reporting maximizes
over the whole image and hence over its closed convex hull.

The proof has four steps. The incentive inequalities first give
local integrability without assuming measurability. They then yield
an additive distributional decomposition on the two nonparallel
angle strips. Legal offsets from the diagonal identify the component
measures and bound them by the full curvature measure. Finally,
monotone sandwiches recover the original point values, the antipodal
traces are matched, and the remaining curvature measure is integrated
to construct a Minkowski complement. Each step is carried out for
the original rule rather than only for an almost-everywhere
representative.

\subsection{Curvature and integrability}

The support inequality evaluated at the support point $z(a)$ gives
\[
 h(a+u)+h(a-u)\geq 2\cos u\,h(a).
\]
Pairing with any nonnegative smooth periodic test, dividing by
$u^2$ and taking the limit proves that
\[
 \mu_L:=D^2h+h\,da
\]
is a positive distribution, hence a finite positive Radon measure.
The representation of a positive distribution as a measure follows
from its order-zero bound on each compact set: if $\eta\geq1$
on that set, positivity bounds its value on a test $\phi$ by
$\|\phi\|_\infty\langle\mu_L,\eta\rangle$.
Strict convexity and Lemma~\ref{lem:exposure} imply $h\in C^1$ and
\[
 z=h n+h't,\qquad Dz=t\,\mu_L.
 \tag{B.1}\label{eq:plane-curvature}
\]
Indeed $Dh'=\mu_L-h\,da$, and the other product-rule terms cancel
using $t'=-n$. Thus $z$ is continuous and of bounded variation.
Its derivative has no atoms: a locally finite $Df=\nu$ has the
primitive $r\mapsto\nu((c,r])$ up to an additive constant, whose
jump at $r$ is $\nu(\{r\})$, and a continuous representative has
no jump. Since $Dz(\{a\})=t(a)\mu_L(\{a\})$ and $|t|=1$, this
makes $\mu_L$ atomless. Continuity is used here, not absolute
continuity.
This is curvature as a measure of the normal-angle variable.
An interval of normals supporting a corner has zero such
curvature; it does not create an atom. Singular-continuous
curvature remains possible.

For a one-coordinate increase $a\to a+u$, $0<u<\pi$, let
$\Delta$ be the change of the actual output. IC gives
\[
 n(a)\cdot\Delta\leq0,\qquad n(a+u)\cdot\Delta\geq0.
\]
In the midpoint frame $m=a+u/2$, these become
\[
 \beta:=t(m)\cdot\Delta\geq0,\qquad
 |n(m)\cdot\Delta|\leq\tan(u/2)\beta.
 \tag{B.2}\label{eq:plane-cone}
\]
At every nonparallel profile choose a linear functional nonzero
on both input tangents and orient each coordinate so its tangent
has a positive value; the two conditions of vanishing on
$t(a_0)$ and on $t(b_0)$ cut out two lines in the dual plane,
which do not cover it. A nearby independent functional has the
same signs. The constants can be chosen once for the whole
neighborhood rather than increment by increment. Let $c_0>0$ be
the least of the four central values $\ell_j(\sigma_it(\cdot))$
and $M=\max_j\|\ell_j\|$. By continuity pick a closed
rectangular neighborhood on which
$\ell_j(\sigma_it(m))\geq c_0/2$ for every midpoint $m$ arising
in it, and shrink its sides below some $\delta<\pi$ with
$M\tan(\delta/2)<c_0/4$. Then for orientation
$\sigma_i\in\{-1,1\}$ and every allowed increment in the box,
\[
 \ell_j(\sigma_i t(m))-\tan(u/2)|\ell_j(n(m))|
 \geq c_0/2-M\tan(\delta/2)>0
 \quad(j=1,2).
\]
Writing $\Delta=\alpha n(m)+\beta t(m)$ and using
\eqref{eq:plane-cone} this implies
$\ell_j(\sigma_i\Delta)\geq
\beta[\ell_j(\sigma_it(m))-\tan(u/2)|\ell_j(n(m))|]\geq0$.
One set of constants thus serves both functionals, both
coordinates and every opponent report in the box.
The scalarizations are bounded since $L$ is bounded.
On an $m$-by-$m$ rectangular grid, summing upper-corner minus
lower-corner oscillations telescopes separately in the two
coordinates; the upper-minus-lower Darboux sums are bounded by
$2\,\operatorname{area}(R)\operatorname{osc}(f)/m$.
Consequently each scalarization is Riemann integrable. It is
also Lebesgue measurable in the completed sense: refining the
grids produces Borel step functions $l_k\uparrow$, $u_k\downarrow$
with $l_k\leq f\leq u_k$ and vanishing integral gap, so their
limits satisfy $l\leq f\leq u$ with $l=u$ almost everywhere, and
$f$ differs from the Borel function $l$ only on a subset of a
Borel null set, which the completed field contains.
Their independence gives joint local integrability of $G$ on the
nonparallel set, without a measurability assumption on the
original rule. The parallel set
$Z=\{(a,b):a-b\in\pi\mathbb Z\}$ is a countable union of lines, hence
Lebesgue null in the plane, and its complement is covered by
countably many such boxes; since the rule takes values in the
bounded set $L$ everywhere, assigning it arbitrary values on
$Z$ preserves completed measurability. Thus
$G\in L^1_{\rm loc}(\R^2)$. This says nothing about the actual
values on $Z$, which are recovered only in
Section~\ref{sub:plane-actual}.

\subsection{Additive decomposition}

For every fixed opponent the scalar support envelope
$c_b(a)$ is Lipschitz, with constant at most
$\sup_{y\in L}\|y\|$. At every differentiability point the inequalities
\[
 c_b(a+u)-c_b(a)\geq
 [n(a+u)-n(a)]\cdot G(a,b)
\]
for positive and negative $u$ give
$c_b'(a)=t(a)\cdot G(a,b)$.
Integration by parts on product boxes and Fubini are now
legitimate by the preceding local integrability. They yield
\[
 n(a)\cdot\partial_aG=0,\qquad
 n(b)\cdot\partial_bG=0
\]
as distributions. Differentiate transversely and use
$\sin(a-b)\ne0$ to invert the smooth two-by-two coefficient
matrix. Thus $\partial_a\partial_bG=0$ on each strip
\[
 \Sigma_+=\{0<a-b<\pi\},\qquad
 \Sigma_-=\{-\pi<a-b<0\}.
\]
It follows that
\[
 G(a,b)=A_\pm(a)+B_\pm(b)
 \quad\hbox{jointly almost everywhere on }\Sigma_\pm,
 \tag{B.3}\label{eq:plane-additive}
\]
with locally integrable one-dimensional representatives.
Here is the required global decomposition argument. On a box,
a distribution $T$ with $\partial_bT=0$ depends only on $a$:
subtract from any test function its $b$-integral times a fixed
smooth $b$-test of mass one. The difference has a compactly
supported $b$-primitive and is annihilated by $T$.
Apply this to $\partial_aG$. The coordinate fibers of each
strip are intervals. A finite chain of overlapping boxes along
a fixed fiber identifies the same one-dimensional distribution
on a common neighborhood of its fixed coordinate. Compatible
local distributions therefore glue on $\R$. Take a
one-dimensional primitive $A$; the remainder has zero
$a$-derivative and similarly equals a distribution $B(b)$.
That these are functions, and not merely distributions,
needs an argument, since a distributional primitive is not
automatically locally integrable. Fix a box and a transverse
test $\psi$ of mass one and put $f(a)=\int G(a,b)\psi(b)\,db$.
Fubini and the local integrability just established give
$f\in L^1_{\rm loc}$, while $Df$ is the same one-dimensional
distribution as $DA$. A distribution on a line with zero
derivative is a constant, because a test of zero total
integral has a compactly supported primitive. Hence $A=f$
up to an additive constant and has a locally integrable
representative; averaging $G-A(a)$ in $a$ does the same for
$B$. Since a locally integrable function embeds injectively
into distributions, $G-A-B$ vanishes jointly almost
everywhere, and a countable cover gives
one joint null exceptional set for \eqref{eq:plane-additive}.
At this stage $A_\pm$ and $B_\pm$ are determined only up to
opposite additive constants, and neither is claimed to be
$2\pi$-periodic.

The weak orthogonality gives scalar distributions
$\alpha_\pm,\beta_\pm$ with
\[
 DA_\pm=t\,\alpha_\pm,\qquad DB_\pm=t\,\beta_\pm.
 \tag{B.4}\label{eq:plane-positive-derivatives}
\]
They are nonnegative. One direct verification retains the
null-set quantifiers. Fix $0<u<\pi$; the exceptional set is
allowed to depend on $u$, since the inequality below is
derived for each $u$ separately and only the limit of the
integrals is taken, so no union over uncountably many null
sets occurs. For almost every $a$,
Fubini permits a $b$ such that both $(a,b)$ and $(a+u,b)$
are in the strip and avoid the two additive exceptional sets.
The common transverse interval is $(a+u-\pi,a)$ in $\Sigma_+$
and $(a+u,a+\pi)$ in $\Sigma_-$, of length $\pi-u>0$.
Equation~\eqref{eq:plane-cone} then gives
$t(a+u/2)\cdot[A(a+u)-A(a)]\geq0$.
After multiplying by a nonnegative compactly supported test
$\phi(a)$, integrating and dividing by $u$, the limit is
\[
 \lim_{u\downarrow0}\int A(r)\cdot
 \frac{\phi(r-u)t(r-u/2)-\phi(r)t(r+u/2)}{u}\,dr
 =\langle t\cdot DA,\phi\rangle\geq0.
\]
The limit is justified by uniform convergence of the smooth
test quotient to $-(\phi t)'$ on a common compact set and local
integrability of $A$.
Varying the second coordinate instead proves $\beta_\pm\geq0$;
this requires its own common interval, namely $(b+u,b+\pi)$ in
$\Sigma_+$ and $(b+u-\pi,b)$ in $\Sigma_-$, again of length
$\pi-u$, after which the same test-quotient limit applies with
$b$ in place of $a$. Positive distributions are Radon measures
by the order-zero bound above, so $A_\pm,B_\pm$ are locally BV
with $|DA_\pm|=\alpha_\pm$ and $|DB_\pm|=\beta_\pm$.
This step does not discard atoms or singular-continuous parts.

\subsection{The diagonal budget}

IC and the option to report the other person's type imply
\[
 0\leq h(a)-n(a)\cdot G(a,b)
 \leq h(a)-n(a)\cdot z(b).
 \tag{B.5}\label{eq:plane-diagonal-gap}
\]
Uniform strict exposure in Lemma~\ref{lem:exposure} shows that
$G(r+\varepsilon,r)\to z(r)$ locally uniformly in $r$ as
$\varepsilon\to0$.

For $\Sigma_+$ use the affine change
$(r,\varepsilon)\mapsto(r+\varepsilon,r)$, whose Jacobian
has absolute value one. Pullback of the joint exceptional
set in \eqref{eq:plane-additive} is null. Fubini, with a
countable exhaustion by compact $r$-intervals, gives a
full-measure set of positive offsets for which
\[
 G(r+\varepsilon,r)=A_+(r+\varepsilon)+B_+(r)
 \quad\text{for almost every }r.
\]
Choose $\varepsilon_k\downarrow0$ in that set.
The left side tends to $z$ locally uniformly, and
$A_+(\cdot+\varepsilon_k)\to A_+$ in $L^1_{\rm loc}$
by translation continuity of $L^1$ functions, which follows
by approximating in $L^1$ by a compactly supported smooth
function, for which the statement is uniform continuity.
On a compact interval $I$ this gives
\[
 \|A_++B_+-z\|_{L^1(I)}
 \leq\|A_+-A_+(\cdot+\varepsilon_k)\|_{L^1(I)}
    +|I|\sup_{r\in I}\|G(r+\varepsilon_k,r)-z(r)\|,
\]
whose right side tends to zero while the left side does not
depend on $k$. Therefore
\[
 A_++B_+=z\quad\text{as an }L^1_{\rm loc}\text{ identity}.
\]
Negative legal offsets give the identical statement for
$\Sigma_-$. That case must repeat the construction with the
exceptional set of its own strip: nothing about $\Sigma_-$
follows from the positive-offset conclusion, and no pairing
between the two full-measure offset sets is needed, nor is
$A_+=A_-$ or $B_+=B_-$ asserted. No joint a.e. equality has
been restricted directly to the null diagonal.

Only now differentiate this identity. From
\eqref{eq:plane-curvature} and \eqref{eq:plane-positive-derivatives},
\[
 \alpha_\pm+\beta_\pm=\mu_L,\qquad
 0\leq\alpha_\pm,\beta_\pm\leq\mu_L.
 \tag{B.6}\label{eq:plane-measure-budget}
\]
All four measures are atomless, being dominated by the
atomless $\mu_L$; the domination is as measures on Borel
sets, not merely between densities. Their BV primitives have
continuous representatives: for a vector measure $m=t\alpha_\pm$
the function $V(r)=m((r_0,r])$, extended by
$V(r)=-m((r,r_0])$ for $r<r_0$, satisfies $DV=m$, is locally
BV, and is continuous precisely because $m$ has no atoms, so
$A_\pm$ equals $V$ up to an additive constant.
Because the continuous sums agree with the continuous $z$
almost everywhere, they agree everywhere.
Singular-continuous parts are allowed and are dominated by
the ambient curvature measure; atomlessness is not absolute
continuity, and no absolute-continuity
conclusion is made for general $L$.

\subsection{Pointwise recovery}
\label{sub:plane-actual}

On a small oriented box the continuous function
$H=A_\pm(a)+B_\pm(b)$ agrees with $G$ almost everywhere.
Write $(u,v)$ for the oriented coordinates and $(u_0,v_0)$
for the target. For small $\delta>0$ the rectangles
\[
 (u_0-\delta,u_0-\tfrac\delta2)\times(v_0-\delta,v_0-\tfrac\delta2),
 \quad
 (u_0+\tfrac\delta2,u_0+\delta)\times(v_0+\tfrac\delta2,v_0+\delta)
\]
lie in the box and have area $\delta^2/4>0$, so each contains
a point $x^\mp_\delta$ off the exceptional set; one such point
serves both functionals at once, since the whole vector
identity $G=H$ holds there. Monotonicity of the two
independent scalarizations in each oriented coordinate gives
$\ell_j\cdot H(x^-_\delta)\leq\ell_j\cdot G(u_0,v_0)
\leq\ell_j\cdot H(x^+_\delta)$ for $j=1,2$.
Letting $\delta\downarrow0$ and using continuity of $H$ and
invertibility of $y\mapsto(\ell_1\cdot y,\ell_2\cdot y)$
identifies the original rule at every nonparallel profile.
Positive area is what excludes
exceptional whole columns, not just isolated points.

At the diagonal, the uniform gap bound
\eqref{eq:plane-diagonal-gap} identifies the trace with
the prescribed value $z$. At the antipodal seam, for each
fixed $b$ the two traces are
\[
 L_+(b)=A_+(b+\pi)+B_+(b),\qquad
 L_-(b)=A_-(b-\pi)+B_-(b).
\]
Their difference $J=L_+-L_-$ is continuous and locally BV.
Closedness of $C_b$, continuity of the traces, and periodicity of its
support function give the explicit seam relations
\[
 L_\pm(b)\in C_b,qquad
 -n(b)\cdot L_\pm(b)=c_b(b+\pi)=c_b(b-\pi).
 \tag{B.7}\label{eq:plane-seam-support}
\]
Thus both traces maximize the same supporting direction $-n(b)$ of
the fixed compact convex menu, and hence
$n(b)\cdot J(b)=0$ at every $b$.
For a translated measure write
$\alpha^{[s]}(E)=\alpha(E+s)$. Then
\[
 DJ=t\,\sigma,\qquad
 \sigma=\beta_+-\beta_--\alpha_+^{[\pi]}
                         +\alpha_-^{[-\pi]}.
 \tag{B.8}\label{eq:plane-seam-measure}
\]
The signs follow from $t(b+\pi)=t(b-\pi)=-t(b)$.
Distributional product differentiation gives
\[
 0=D(n\cdot J)=(t\cdot J)\,db+n\cdot(t\,\sigma)
               =(t\cdot J)\,db .
\]
Continuity implies $t\cdot J=0$ everywhere, so $J=0$.
This calculation retains any singular-continuous part of
$\sigma$; it vanishes in the normal component because
$n\cdot t$ vanishes identically.
The periodicity used here is that of the rule. The
decomposition functions themselves may have opposite
additive increments over one period and are not assumed
separately periodic.

Define $g_b$ using the actual values off the seams and the
common traces at the seams. It is a continuous support
selection for $C_b$ on the full circle. Membership and the
support equality at a seam follow by taking limits in the
closed set $C_b$. The support inequalities give
\[
 [n(a+u)-n(a)]\cdot g_b(a)
 \leq c_b(a+u)-c_b(a)
 \leq[n(a+u)-n(a)]\cdot g_b(a+u).
\]
Taking the two-sided difference quotient shows
$c_b'(a)=t(a)\cdot g_b(a)$ everywhere.
If $y\in E_{C_b}(n(a))$, the same support inequality with
the fixed maximizer $y$ gives
$n(a)\cdot y=c_b(a)$ and $t(a)\cdot y=c_b'(a)$.
These two planar coordinates force $y=g_b(a)$.
Thus the exposed face at every direction is a singleton,
and the original IC selection at the seam, which lies in
$C_b$ and maximizes the seam direction, must equal
$g_b(a)$ as well. Each original column is now continuous
and has compact image; a compact set in the plane has
compact convex hull, being the continuous image of
$M_b^3\times\Delta_2$, so its initial closed convex hull
equals its ordinary convex hull. For the row conclusions
apply everything above to $\widetilde G(a,b)=G(b,a)$:
its first-coordinate IC is the second-coordinate IC of $G$,
its diagonal is unchanged, and its values lie in the same
$L$, so the hypotheses hold verbatim. This exchanges the
roles of the two strips and of the additive functions;
it does not assert anonymity or that $\widetilde G=G$.

\subsection{Minkowski complements}

On the two open arcs of a fixed column,
$Dg_b=t\,\nu_b$ with $0\leq\nu_b\leq\mu_L$ by
\eqref{eq:plane-measure-budget}. Each arc has finite
variation up to its endpoints, and gluing two such arcs at
two junction points produces a distributional derivative
equal to the two open-arc derivatives together with atoms
carrying the right-minus-left traces at the junctions. Those
traces were just shown to agree, so the atoms vanish, and
$g_b$ is of bounded variation on the full circle; continuity
alone would not give this. Hence
\[
 c_b''+c_b\,da=\nu_b,\qquad
 (h-c_b)''+(h-c_b)\,da=\mu_L-\nu_b\geq0.
 \tag{B.9}\label{eq:plane-complement}
\]
For completeness the planar support criterion is verified
directly. Put $d_b=h-c_b$, a periodic $C^1$ function, and
$q_b=d_b n+d_b't$. Its distributional derivative is
$t(\mu_L-\nu_b)$ and it is continuous. Lift two circle directions to $\alpha,\beta$ with
$|\beta-\alpha|\leq\pi$. If $0\leq\beta-\alpha\leq\pi$, then
\[
 n(\beta)\cdot[q_b(\beta)-q_b(\alpha)]
 =\int_{(\alpha,\beta]}\sin(\beta-r)\,(\mu_L-\nu_b)(dr)\geq0,
\]
since $0\leq\beta-r\leq\pi$ on the domain of integration.
For the other orientation, $0\leq\alpha-\beta\leq\pi$, the
integral runs over $(\beta,\alpha]$ and carries a minus sign,
while $\sin(\beta-r)\leq0$ there because
$-\pi\leq\beta-r\leq0$; the two sign changes cancel and the
same conclusion holds. Both signs must be tracked: the
forward formula cannot simply be reused. Coincident
directions give equality, and at an arc of length exactly
$\pi$ the endpoint sine vanishes. Hence for all circle
directions
$n(\beta)\cdot q_b(\alpha)\leq d_b(\beta)$, with equality
at $\alpha=\beta$. The compact convex hull $D_b$ of
$q_b(\Circ)$ has support function $d_b$. Additivity and
uniqueness of support functions give
\[
 L=C_b+D_b
\]
for every actual column $b$, not merely almost every one.

In the unit-disk case $\mu_L=da$, so
$Dg_b=t\,\lambda_b(a)\,da$, $0\leq\lambda_b\leq1$.
Every angular section is absolutely continuous and
one-Lipschitz. These estimates hold for every fixed
opponent and both input positions, and imply joint
continuity by changing the coordinates in sequence,
with the bound $d_{\Circ}(a,a')+d_{\Circ}(b,b')$ in the
shorter arc distance.
For general $L$ put
$\omega_\mu(\delta)=\sup\{\mu_L(I):I\text{ an arc of length}
\leq\delta\}$. Then $\omega_\mu(\delta)\to0$: otherwise there
are $\delta_k\downarrow0$ and arcs $I_k$ with
$\mu_L(I_k)\geq\eta>0$, and a convergent subsequence of their
midpoints gives a point $a_*$ such that every closed arc
centered at $a_*$ carries mass at least $\eta$, whence
$\mu_L(\{a_*\})\geq\eta$ by continuity from above,
contradicting atomlessness. Since every actual section
satisfies $|Dg|=\nu\leq\mu_L$, all of them share this single
modulus, and
\[
 \|G(a,b)-G(a',b')\|\leq
 \omega_\mu(d_{\Circ}(a,a'))+\omega_\mu(d_{\Circ}(b,b')),
\]
which is the joint continuity asserted, including near the
diagonal and the seam. Pointwise atomlessness alone would not
give this uniformity. The modulus need not be linear in
$\delta$: singular-continuous variation
need not disappear, and no Lipschitz or absolute-continuity
upgrade is obtained. This proves
Lemma~\ref{lem:planar-body}.
\section{Joint envelopes}\label{app:transport}

The analytic argument treats a binary IC diagonal-unanimous
vector rule with values in $K$. Boundary-valuedness will only be
used after the global menu decomposition has been established.
Extend it homogeneously in each positive radial coordinate to
$\Omega=(\R^d\setminus\{0\})^2$ and retain the notation $F$.
Define $C_Q=\overline{\operatorname{conv}}\{F(P,Q):P\ne0\}$
and the analogous closed convex row menu $R_P$. They are compact
subsets of $K$ even when the original image is not closed. Put
\[
 U(P,Q)=P\cdot F(P,Q)=h_{C_Q}(P),\qquad
 V(P,Q)=Q\cdot F(P,Q)=h_{R_P}(Q).
\]
All equalities are pointwise consequences of IC.

\subsection{Faces and measurability}

We first establish the cross-coordinate constancy of projected
exposed faces. This is the geometric input that allows scalar support
envelopes to be continuous even before the vector rule is known to be
measurable.

For any two-plane $L$ containing $Q$, restriction of both reports
to $L$ and orthogonal projection of outcomes into $K_L$ preserves
both utility inequalities and the diagonal support point.
Lemma~\ref{lem:planar-body} applies to this body-valued rule.
Its restricted column's closed convex hull equals $\Pi_LC_Q$:
for each direction $u\in L$, the truthful report $u$ already
maximizes the entire ambient menu, so the two sets have the
same support function on $L$. Consequently
\[
 K_L=\Pi_LC_Q+D_{Q,L}.
 \tag{C.1}\label{eq:plane-projected-summand}
\]
Since $K_L$ is strictly convex, each exposed face of
$\Pi_LC_Q$ is a singleton: the exposed face of a Minkowski
sum is the sum of the exposed faces, and adding a point of
the nonempty $E_{D_{Q,L}}(u)$ to two distinct points of
$E_{\Pi_LC_Q}(u)$ would give two distinct points of
$E_{K_L}(u)$.
If $P,Q$ are independent, take $L=\operatorname{span}(P,Q)$.
Two points of $E_{C_Q}(P)$ have the same $P$-inner product
and the same support value, so their projections lie in
$E_{\Pi_LC_Q}(P)$ and coincide; their difference lies in
$L^\perp$, and since $Q\in L$ their $Q$-inner products agree.
It is this cross-coordinate constancy that is used below;
constancy of the $P$-inner product on $E_{C_Q}(P)$ is the
definition of an exposed face and carries no information.
If $P=\pm\lambda Q$ for $\lambda>0$, use every two-plane
containing that axis. The difference of two points of the face
lies in all their orthogonal complements; a vector orthogonal
to every such plane is orthogonal to all of $\R^d$, since any
$v$ lies in $\operatorname{span}(Q,v)$, one of these planes.
The intersection of the complements is therefore zero and the
ambient exposed face is itself a
singleton. For $d=2$ the only such plane is the whole space.
Rows satisfy the interchanged assertion, obtained by applying
the column statement to $\widetilde F(P,Q)=F(Q,P)$.

We next pass from this face constancy to envelope continuity and
measurability. The order is essential: the joint Borel set of
differentiability points is constructed from the already continuous
envelope, and only then is Tonelli used.

Fix $P$ and let $Q_j\to Q$. The outcomes $F(P,Q_j)$ belong
to the same compact convexified row menu. Any cluster point is in its
$Q$-exposed face by the row support inequalities and continuity
of a support function. The preceding face property makes
the $P$-inner product constant on that face. Thus
$U(P,Q_j)\to U(P,Q)$. In its first variable $U$ has a common
Lipschitz constant $R=\max_{y\in K}\|y\|$. The estimate
\[
 |U(P_j,Q_j)-U(P,Q)|
 \leq R\|P_j-P\|+|U(P,Q_j)-U(P,Q)|
\]
proves joint continuity. The argument for $V$ is the mirror
image and is worth stating, because it uses the column face
rather than the row face: fix $Q$, let $P_j\to P$, and note
that the outcomes lie in the compact $C_Q$, so any cluster
point belongs to $E_{C_Q}(P)$ by first-coordinate support
maximization; the face property just proved keeps the
$Q$-inner product constant there, giving
$V(P_j,Q)\to V(P,Q)$, and the same Lipschitz constant $R$ in
the second variable completes the estimate.
Only one inner product of the outcome is controlled at a
time; cluster points are never identified with one another,
so continuity of the
original vector selection has not been proved.

For a continuous function convex in $P$, its differentiability
set in $P$ is jointly Borel. On local boxes this can be
expressed through the limits of countably many rational
one-sided coordinate difference quotients, each of which is
jointly continuous by the joint continuity of $U$, so their
one-sided limits $d_i^\pm$ are Borel and the set where
$d_i^+=d_i^-$ for every $i$ is Borel. This coordinatewise
condition does characterize differentiability here, although
it would not for a general function. Any $y\in E_{C_Q}(P)$
satisfies $d_i^-\leq y_i\leq d_i^+$ by the support
inequalities; equality of the one-sided limits therefore
makes the face a single point $y_*$, and for $y_v\in
E_{C_Q}(P+v)$ compactness and continuity of support values
give $y_v\to y_*$, whence
\[
 0\leq h_{C_Q}(P+v)-h_{C_Q}(P)-v\cdot y_*
 \leq\|v\|\,\|y_v-y_*\|=o(\|v\|),
\]
which is full differentiability with gradient $y_*$.
Each fixed-$Q$ exception
is Lebesgue null because $h_{C_Q}$ is $R$-Lipschitz on $\R^d$.
Since the joint exception is already known to be Borel,
Tonelli applied to its indicator makes it jointly null;
the order matters, as Fubini could not be applied to a set
not yet known to be measurable.
At differentiability points support maximization gives
$F(P,Q)=\nabla_PU(P,Q)$. The latter is a jointly Borel
function where defined. Arbitrary values of $F$ on the
joint null exception make the original rule measurable
in the completed Lebesgue sigma-field and locally bounded.
Interchanging the two coordinates gives
$F(P,Q)=\nabla_QV(P,Q)$ almost everywhere as well. More explicitly,
for fixed $P$ the second-coordinate IC condition makes $V(P,\cdot)$
the support function of the row menu, so its gradient at every
differentiability point is the unique selected support maximizer;
the same Borel-set and Tonelli argument supplies one joint null set.
This is the precise measurability used below; no Borel
selection hypothesis has been added.

\subsection{Matrix measures}

Define on $\Omega$, in fixed Euclidean coordinates,
\[
 A_{ij}=\partial_{P_j}F_i,\qquad B_{ij}=\partial_{Q_j}F_i.
\]
The preceding a.e. identities and Fubini give
\[
 A=D^2_{PP}U=A^\top\succeq0,\qquad
 B=D^2_{QQ}V=B^\top\succeq0
 \tag{C.2}\label{eq:positive-matrices}
\]
as distributions. For each constant vector $v$, the scalar
distribution $v^\top Av$ is positive, hence a Radon measure.
Polarization gives every entry as a signed Radon measure.
Local finiteness of the positive trace follows directly:
for a compact set choose a nonnegative smooth cutoff $\chi$
equal to one on it, and use
\[
 \langle\operatorname{tr}A,\chi\rangle
       =-\int_\Omega F\cdot\nabla_P\chi.
\]
Boundedness of $F$ bounds the right side. Entries are controlled
by the trace: for a relatively compact Borel $E$ the PSD matrix
$A(E)$ has $|A_{ij}(E)|\leq\sqrt{A_{ii}(E)A_{jj}(E)}
\leq\frac12[A_{ii}(E)+A_{jj}(E)]$ by its two-by-two minors, and
applying this to each piece of a finite Borel partition and
taking the supremum over partitions bounds the \emph{total
variation}
\[
 |A_{ij}|(E)\leq\tfrac12[A_{ii}(E)+A_{jj}(E)]
 \leq(\operatorname{tr}A)(E).
\]
The distinction matters: vanishing net mass $A_{ij}(E)=0$ would
not by itself exclude cancellation, whereas vanishing total
variation does.
The same facts hold for $B$. Hessians may have singular parts.

Let $T_{ijk}=\partial_{Q_k}\partial_{P_j}F_i$.
Symmetry of $A$ interchanges $i,j$; symmetry of $B$ and
commutation of derivatives interchange $i,k$.
Thus all three indices are symmetric as distributions.
Positive homogeneity of degree zero in $P$ gives
$\sum_\ell P_\ell\partial_{P_\ell}F_i=0$ in the weak sense,
which must be derived rather than read off a classical chain
rule, $F$ not being known to be differentiable. For a test
$\phi$ the integral $I(s)=\int F_i(sP,Q)\phi\,dP\,dQ$ is
constant in $s$; substituting $R=sP$ gives
$I(s)=\int F_i(R,Q)s^{-d}\phi(R/s,Q)\,dR\,dQ$, which may be
differentiated at $s=1$ because the test functions involved
have a common compact support and $F_i$ is locally bounded.
This yields
$0=-\int F_i[d\phi+P\cdot\nabla_P\phi]
=\langle\sum_\ell P_\ell\partial_{P_\ell}F_i,\phi\rangle$,
the factor $d$ coming from the radial Jacobian.
Differentiate in $Q_j$, noting that $P_\ell$ does not depend
on $Q$ so no coefficient term appears, and use the mixed
symmetry to obtain
\[
 \sum_\ell P_\ell T_{i\ell j}
 =\sum_\ell P_\ell T_{ij\ell}=0,\qquad
 (P\cdot\nabla_Q)A_{ij}=0.
 \tag{C.3}\label{eq:transport-equation}
\]
Only smooth coordinate multipliers and distributional
derivatives occur; no product of singular distributions
is taken.

\subsection{Transport}

Set $\mathcal O=\{(P,Q)\in\Omega:P,Q\text{ independent}\}$.
The flow $S_t(P,Q)=(P,Q+tP)$ stays in this set for every
real $t$, since $Q+tP$ can be neither zero nor a multiple of
$P$ when $P,Q$ are independent; it is a diffeomorphism with
inverse $S_{-t}$ and unit Jacobian. Its generating field
$Y=(0,P)$ has zero divergence, as $\operatorname{div}_QP=0$.
Writing $\phi_t(P,Q)=\phi(P,Q-tP)$, the pullback of a scalar
entry $M$ pairs as $\langle S_t^*M,\phi\rangle
=\langle M,\phi_t\rangle$, and over a bounded interval of $t$
these transported test functions have supports whose union is
a compact subset of $\mathcal O$, so the parameter derivative
may be taken in the test-function topology:
\[
 \frac{d}{dt}\langle S_t^*M,\phi\rangle
 =\langle M,-P\cdot\nabla_Q\phi_t\rangle
 =\langle (P\cdot\nabla_Q)M,\phi_t\rangle=0
\]
by \eqref{eq:transport-equation}, the middle equality being
exactly where the vanishing divergence is used. Hence
$S_t^*A=A$; no derivative of a singular measure along the
flow is taken pointwise.
Degree-zero homogeneity in $Q$, obtained by differentiating
the original radial identity for $F$ in $P$, also gives
$R_\varepsilon^*A=A$ for
$R_\varepsilon(P,Q)=(P,\varepsilon Q)$, $\varepsilon>0$.
Indeed, $F\circ R_\varepsilon=F$ pointwise by radial homogeneity.
Taking the weak derivative with respect to $P$ on both sides gives
$D_P(F\circ R_\varepsilon)=D_PF$; since the first component of
$R_\varepsilon$ is $P$, the left side is exactly the componentwise
distributional pullback $R_\varepsilon^*A$, including the inverse
Jacobian in the $Q$ variables.
The composition
\[
 \Phi_\varepsilon=S_1\circ R_\varepsilon,\qquad
 \Phi_\varepsilon(P,Q)=(P,P+\varepsilon Q)
\]
therefore satisfies $\Phi_\varepsilon^*A=A$ on $\mathcal O$.

The star here is the scalar distribution pullback of each
fixed-coordinate matrix entry, not the pullback of a
covariant two-tensor. Its exact convention is
\[
 \langle\Phi_\varepsilon^*M,\varphi\rangle
 =\left\langle M,\varepsilon^{-d}
       \varphi\!\left(P,\frac{Q-P}{\varepsilon}\right)\right\rangle .
 \tag{C.4}\label{eq:pullback-convention}
\]
The inverse Jacobian is $\varepsilon^{-d}$.
For function densities this is composition; for positive
matrix measures it preserves positivity against nonnegative
tests. The shear has Jacobian one, and radial invariance
includes this inverse-Jacobian factor rather than dropping it.
This proves Lemma~\ref{lem:transport}.
\section{Global menu decomposition}\label{app:budget}

\subsection{The diagonal budget}

Put $F_\varepsilon(P,Q)=F(P,P+\varepsilon Q)$ on $\mathcal O$;
this stays in $\mathcal O$ for $\varepsilon>0$.
On any compact subset of this open set, $\|P\|$ is bounded below
by some $m>0$ and $\|Q\|$ above by some $M$. For nonzero
$x,y$ one has
$\|x/\|x\|-y/\|y\|\,\|\leq2\|x-y\|/\|y\|$, so
\[
 \sup\bigl\|\widehat{P+\varepsilon Q}-\widehat P\bigr\|
 \leq2\varepsilon M/m\longrightarrow0 .
\]
Since $F$ is radially homogeneous,
$F_\varepsilon(P,Q)=F(\widehat P,\widehat{P+\varepsilon Q})$,
and the uniform exposure modulus of
Lemma~\ref{lem:exposure} applies pointwise to give
\[
 \|F_\varepsilon(P,Q)-X(P)\|\leq\omega_K(2\varepsilon M/m),
 \qquad
 F_\varepsilon(P,Q)\longrightarrow X(P)
\]
locally uniformly as $\varepsilon\downarrow0$. This is a
pointwise incentive-loss bound on the actual rule, not a
continuity or almost-everywhere statement, and the limiting
profile $(P,P)$ need not lie in $\mathcal O$.
The linear weak chain rule and the componentwise pullback
convention \eqref{eq:pullback-convention} give
\[
 D_PF_\varepsilon
   =\Phi_\varepsilon^*A+\Phi_\varepsilon^*B
   =A+\Phi_\varepsilon^*B.
 \tag{D.1}\label{eq:diagonal-chain}
\]
The second coefficient is one, not $\varepsilon$, because
the derivative of the second input $P+\varepsilon Q$ with
respect to $P$ is the identity matrix; $\varepsilon$ appears
only on differentiating in the new variable $Q$.
Since the chain rule is applied to an $L^1_{\rm loc}$ function,
it is verified on tests: with
$\psi(R,S)=\varepsilon^{-d}\phi(R,(S-R)/\varepsilon)$ one has
\[
 \partial_{R_j}\psi+\partial_{S_j}\psi
 =\varepsilon^{-d}(\partial_{P_j}\phi)\circ\Phi_\varepsilon^{-1},
\]
the two terms in $\varepsilon^{-1}\partial_{Q_j}\phi$
cancelling, and pairing $-F_i$ with each side gives
\eqref{eq:diagonal-chain}. Positive pullback
preserves the PSD property of $B$, as
$\langle v^\top\Phi_\varepsilon^*Bv,\varphi\rangle
=\langle v^\top Bv,\varepsilon^{-d}\varphi\circ
\Phi_\varepsilon^{-1}\rangle\geq0$ for nonnegative $\varphi$,
with no density needed.
Local uniform convergence implies convergence of the weak
derivatives against every compactly supported test, since
$|\langle\partial_{P_j}(F_\varepsilon)_i-\partial_{P_j}X_i,
\varphi\rangle|
\leq\sup_{\operatorname{supp}\varphi}\|F_\varepsilon-X\|
\int|\partial_{P_j}\varphi|$.
The limit is the product measure it is written as: for any
test, Fubini in $Q$ gives
$\langle\partial_{P_j}X_i,\varphi\rangle
=\langle(D^2h_K)_{ij},\int\varphi\,dQ\rangle$, the inner
integral being a smooth compactly supported $P$-test away
from the origin. Convergence of distributions, not of
classical derivatives, is what is claimed.
Consequently, for every constant $v\in\R^d$ and nonnegative
test $\varphi$ on $\mathcal O$,
\[
 \left\langle v^\top(\mathsf K-A)v,\varphi\right\rangle
 =\lim_{\varepsilon\downarrow0}
     \left\langle v^\top\Phi_\varepsilon^*Bv,\varphi\right\rangle
 \geq0,
 \qquad
 \mathsf K=(D^2h_K)\otimes\mathcal L_Q^d .
\]
Thus the full matrix-measure inequality is
\[
 0\preceq A\preceq (D^2h_K)\otimes\mathcal L_Q^d
 \quad\text{on }\mathcal O.
 \tag{D.2}\label{eq:loewner-budget}
\]
The Hessian of the finite convex support function is a
PSD matrix Radon measure on the punctured $P$-space.
Equation~\eqref{eq:loewner-budget} retains its absolutely
continuous and singular parts. It is not a density-only
a.e. inequality.

\subsection{Parallel reports}

For $d\geq3$ let
$Z=\Omega\setminus\mathcal O=\{(P,Q):P\parallel Q\}$.
Fix a compact set $L_0\Subset\Omega$ and a nonnegative
smooth compactly supported $\psi$ equal to one near it.
On this support define
\[
 a=\frac{P\cdot Q}{\|P\|^2},\qquad W=Q-aP,\qquad r^2=\|W\|^2.
\]
Direct differentiation gives $\nabla_Pr^2=-2aW$.
Choose a nonnegative smooth $\theta$ equal to one on $[0,1]$
and zero on $[4,\infty)$, and set
$\chi_\delta=\psi\theta(r^2/\delta^2)$.
This is a smooth test, including at $W=0$, and equals
one on $Z\cap L_0$. The support of its cutoff factor
has volume $O(\delta^{d-1})$: for every fixed $P$,
decompose bounded $Q$ into its component on the line
$\R P$ and its $d-1$ orthogonal components.
The cutoff derivative in $P$ is $O(\delta^{-1})$,
since $a$ is bounded and $|W|\leq2\delta$ where it is
nonzero. All constants depend only on the fixed compact
neighborhood and $K$.

Because $\operatorname{tr}A$ is a positive measure and
$F$ is bounded, integration by parts gives
\[
 \begin{split}
 0\leq(\operatorname{tr}A)(Z\cap L_0)
 &\leq\langle\operatorname{tr}A,\chi_\delta\rangle\\
 &=-\int F\cdot\nabla_P\chi_\delta
 \leq C_1\delta^{d-2}+C_2\delta^{d-1}.
 \end{split}
 \tag{D.3}\label{eq:parallel-cutoff}
\]
The two orders of magnitude come from the $d-1$ orthogonal
$Q$-directions and one differentiation of the cutoff, the
first term of $\nabla_P\chi_\delta$ contributing
$O(\delta^{d-1})$ and the term carrying $\theta'$, supported
where $\delta\leq\|W\|\leq2\delta$, contributing
$O(\delta^{d-2})$.
The last estimate bounds the absolute value of the
integral. It tends to zero for $d\geq3$, and a countable
compact exhaustion gives $(\operatorname{tr}A)(Z)=0$. The
total-variation bound $|A_{ij}|\leq\operatorname{tr}A$
established above then makes every entry vanish on $Z$ in
total variation, not merely in net mass.
In dimension two the leading term is $O(1)$ and the estimate
gives nothing; nor may it be replaced by the observation that
$Z$ is Lebesgue null, since $A$ may carry singular mass.
No conclusion from \eqref{eq:parallel-cutoff}
is used in dimension two.

Let $\lambda=\operatorname{tr}(D^2h_K)$, which may
itself be singular in $P$. Tonelli gives, locally,
\[
 (\lambda\otimes\mathcal L_Q^d)(Z)
 =\int \mathcal L^d\{Q:Q\parallel P\}\,\lambda(dP)=0,
\]
because the inner set is a line and $d\geq2$; the null set
is produced in the $Q$-variable, so no absolute continuity of
$\lambda$ is needed. $Z$ is closed in $\Omega$, being cut out
by the equations $P_iQ_j-P_jQ_i=0$, so the indicator is Borel
and Tonelli applies.
Therefore $\mathsf K(Z)=0$ in total variation as well, and
since both matrix measures give no mass to $Z$, integrating
$v^\top(\mathsf K-A)v$ against a nonnegative test may ignore
$Z$ and reduces to the inequality already proved on
$\mathcal O$. Hence
\eqref{eq:loewner-budget} extends to all of $\Omega$
when $d\geq3$. Again this rests on two separate removals of
mass, not on the Lebesgue nullity of $Z$.

\subsection{Minkowski complements}

Write $w(P,Q)=h_K(P)-U(P,Q)$. It is jointly continuous,
positively homogeneous of degree one in $P$, and
$|w(P,Q)|\leq2R\|P\|$ with $R=\max_{y\in K}\|y\|$.
For $d\geq3$ the preceding budget says
$D^2_{PP}w=\mathsf K-A\succeq0$ on $\Omega$.
For any fixed $v$ and nonnegative
$\varphi\in C_c^\infty(\R^d\setminus\{0\})$, the scalar
function
\[
 g(Q)=\int w(P,Q)\,\partial_v^2\varphi(P)\,dP
\]
is continuous. Joint positivity makes
$\int g(Q)\eta(Q)\,dQ\geq0$ for every nonnegative
$Q$-test $\eta$. If $g$ were negative at any actual
$Q$, continuity would give a negative neighborhood
and a test contradicting this inequality.
Thus $g(Q)\geq0$ for every $Q$. This reasoning applies
to each test and direction; no common exceptional
set over uncountably many tests is required.

For every fixed $Q$, $w(\cdot,Q)$ consequently has a
PSD distributional Hessian on the punctured space.
Local mollification on a ball avoiding zero gives
smooth functions with PSD Hessians; local uniform
convergence gives local convexity of $w$.
On a line segment avoiding zero, local convexity
on its compact parameter interval gives the
convexity inequality along the entire segment.
The punctured space is not convex, so this needs the
usual chaining: cover the compact parameter interval by
finitely many intervals of local convexity, let $\ell>0$ be
a Lebesgue number of the cover, and for every sufficiently large
$m$ with $2/m<\ell$,
so that each triple of consecutive grid points
$(i-1)/m,i/m,(i+1)/m$ lies in one such interval. The
resulting grid increments are nondecreasing, whence
$F_0(i/m)\leq(1-i/m)F_0(0)+(i/m)F_0(1)$ for
$F_0(t)=w((1-t)x+ty,Q)$. Given any $t\in[0,1]$, choose grid
indices $i_m$ with $i_m/m\to t$ along such sufficiently large
$m$. Letting $m\to\infty$ and using continuity of $F_0$ gives
$F_0(t)\leq(1-t)F_0(0)+tF_0(1)$ at every $t$.
Set $w(0,Q)=0$. The linear bound above gives
continuity at zero. A segment through zero can be
translated in a direction transverse to its line
because $d\geq2$; every point of the translated segment then
has a nonzero transverse component and avoids the origin.
Apply convexity on the translated
segments and let the translation tend to zero, which also
covers a segment with one endpoint at the origin.
Thus $w(\cdot,Q)$ is finite, continuous, convex
and positively homogeneous on all of $\R^d$.

Such a function is the support function of a
nonempty compact convex set. Its supporting
subgradients exist at every point, and not merely at
almost every one, which is what realizing the support value
in every direction requires. Write $f=w(\cdot,Q)$ and let
$E=\{(y,s):s\geq f(y)\}$, a closed convex set with nonempty
interior having $(P,f(P))$ on its boundary. For
$x_k=(P,f(P)-1/k)$ let $z_k$ be the nearest point of $E$,
which exists because the minimization may be restricted to a
bounded set, and put $v_k=(x_k-z_k)/\|x_k-z_k\|$. Convexity
gives $v_k\cdot(z-z_k)\leq0$ for all $z\in E$, and a
subsequential limit $(a,b)$ of the unit vectors $v_k$
satisfies $a\cdot(y-P)+b(s-f(P))\leq0$ on $E$. Letting
$s\to+\infty$ gives $b\leq0$; if $b=0$ then
$a\cdot(y-P)\leq0$ for all $y$, forcing $a=0$ and
contradicting $\|(a,b)\|=1$. So $b<0$, and $z=-a/b$
satisfies $f(y)\geq f(P)+z\cdot(y-P)$ for all $y$.
Homogeneity then
makes such a subgradient $z$ at $P$ satisfy
$z\cdot P=w(P,Q)$, by applying the subgradient inequality at
$y=0$ and $y=2P$, and hence $z\cdot y\leq w(y,Q)$
for every $y$. The intersection
$D_Q=\{z:z\cdot y\leq w(y,Q)\ \forall y\}$
is closed, nonempty and bounded, since $z\ne0$ gives
$\|z\|\leq w(z/\|z\|,Q)\leq2R$; these supporting subgradients
give $h_{D_Q}=w(\cdot,Q)$. A one-point or lower-dimensional
$D_Q$ is permitted. Therefore
\[
 K=C_Q+D_Q\qquad\text{for every }Q\ne0.
 \tag{D.4}\label{eq:all-column-summand}
\]
The row statement of Proposition~\ref{prop:minkowski} is
obtained by running the whole argument for
$\widetilde F(P,Q)=F(Q,P)$, which is again a body-valued
binary IC diagonally unanimous rule into the same $K$, so
every hypothesis used above holds verbatim; anonymity of $F$
is not required.
When $d=2$, the entire environment is already a
plane, the original rule itself satisfies the hypotheses of
Lemma~\ref{lem:planar-body}, and
\eqref{eq:all-column-summand} follows
directly from it for every actual column and, after the
same exchange of labels, for every actual row. This is a
genuinely separate branch, not the $d\geq3$ conclusion
evaluated at $d=2$.
This proves Proposition~\ref{prop:minkowski}.
It uses continuity of the envelopes, not assumed
continuity of the original rule.
\section{Consequences}\label{app:population}

\subsection{Efficiency and dimension one}

The support map is onto \(\partial K\), so diagonal unanimity
implies that the rule is onto. Conversely, suppose the rule
is onto and IC. For a type \(p\), choose a profile attaining
\(X(p)\). Starting at the diagonal profile \(p_N\), change
reports successively to this profile. Before each change
the relevant agent still truthfully reports \(p\).
To make the chain explicit, if the target profile is
$r=(r_1,\ldots,r_n)$, let $s^0=p_N$ and
$s^k=(r_1,\ldots,r_k,p,\ldots,p)$ for $k=1,\ldots,n$.
Agent $k$ has true type $p$ at $s^{k-1}$, so IC gives
$p\cdot F(s^{k-1})\geq p\cdot F(s^k)$. Chaining these
inequalities and using $F(s^n)=X(p)$ gives
\[
                       p\cdot F(p_N)\geq p\cdot X(p)=h_K(p).
\]
Strict exposure forces \(F(p_N)=X(p)\).

At a diagonal profile, every outcome other than \(X(p)\)
is strictly worse for every participant, and \(X(p)\) is
itself feasible. Weak Pareto
efficiency there therefore implies diagonal unanimity.
Conversely, a fixed peak rule is Pareto efficient:
any change from its controlling agent's unique optimum
makes that agent strictly worse off, so no boundary point
can weakly improve on the outcome for everyone.
The remaining two links are immediate but should be stated,
since the corollary asserts a full equivalence. Because the
population is nonempty, an outcome strictly improving for
everyone also weakly improves for everyone and strictly for
one, so Pareto efficiency implies weak Pareto efficiency;
and efficiency at all profiles implies efficiency at the
diagonal ones. With the equivalence of unanimity and
ontoness proved above, the implications close into the cycle
\[
 \begin{gathered}
 \text{fixed dictatorship}\Rightarrow\text{Pareto}
 \Rightarrow\text{weak Pareto}\\
 \Rightarrow\text{diagonal weak Pareto}
 \Rightarrow\text{unanimity}
 \Rightarrow\text{fixed dictatorship},
 \end{gathered}
\]
the last arrow being Theorem~\ref{thm:frontier} and hence
available only within the IC class.
This proves Corollary~\ref{cor:efficiency}.

In dimension one take \(K=[-1,1]\), with report and outcome
space \(\{-1,1\}\). The two-agent rule choosing \(1\) whenever
at least one report is \(1\), and \(-1\) otherwise, is
unanimous. A true \(1\) type obtains its top. A true
\(-1\) type either obtains its top when the other report
is \(-1\), or cannot alter the outcome when it is \(1\).
Thus the rule is ordinarily IC. At the two disagreement
profiles its choice is \(1\), which rules out either
fixed label. This establishes the comparison stated
after Theorem~\ref{thm:frontier}.

\subsection{Continuous rules}\label{app:continuous}

For completeness, assume in this subsection that the rule
is continuous as well as IC and diagonally unanimous.
Choose an interior ball \(\overline B(y_0,\rho)\subset K\),
\(\rho>0\). Radial projection from \(y_0\) is a homeomorphism
\(r:\partial K\to S^{d-1}\): every ray meets the boundary
once, and compactness gives continuity of the bijection
and its inverse. Put \(T(p)=r(X(p))\).
The interior ball gives
\[
                       p\cdot(X(p)-y_0)\geq\rho>0,
\]
so \(p\cdot T(p)>0\). The normalized straight-line homotopy
from \(T(p)\) to \(p\) never vanishes. Consequently \(T\)
has degree one. This construction does not require
injectivity of the support map at a nonsmooth point.

Let \(m=d-1\geq1\) and \(g=r\circ F:(S^m)^n\to S^m\).
Then \(g\circ\Delta=T\), where \(\Delta\) is the diagonal
map. In degree \(m\), the integer cohomology of the
product is the direct sum of the \(n\) coordinate
copies of \(\mathbb Z\). Writing
\[
                       g^*\omega=\sum_i a_i\pi_i^*\omega
\]
for a generator \(\omega\), diagonal pullback gives
\(\sum_i a_i=1\). Some coordinate therefore has
nonzero coefficient, say $a_i\ne0$. For fixed opponents' reports
$p_{-i}$, let
\[
 \iota_{i,p_{-i}}:S^m\longrightarrow(S^m)^n
\]
be the coordinate inclusion that varies only report $i$. Pulling
the cohomology identity back by this inclusion gives
\[
 (g\circ\iota_{i,p_{-i}})^*\omega=a_i\omega.
\]
Thus every such slice has degree $a_i$. The value is independent
of $p_{-i}$ because the space of opponents' reports is path
connected, so any two slices are homotopic through slices.

A map from $S^m$ to itself with nonzero degree is onto. Indeed, if
its image omitted a point, it would factor through the punctured
sphere, which is contractible, and the induced map on top-degree
cohomology would be zero. Hence every $i$-slice of $g$, and therefore
every corresponding $i$-slice of $F$, is onto the whole boundary.
Fix a true type $p_i$. Since that slice can attain $X(p_i)$, IC says
the truthful report yields utility at least $p_i\cdot X(p_i)$.
This is the maximum of $p_i\cdot y$ on $K$, so equality holds and
strict convexity makes the chosen point uniquely $X(p_i)$. The same
coordinate $i$ therefore selects its peak for every opponents'
profile and is a fixed dictator.

This is the standard continuous finite-dimensional
benchmark described in the introduction.
The main proof does not assume this continuity;
it derives the needed regularity and menu structure
from the incentive inequalities themselves.
\end{document}